\documentclass[11pt]{article}
\usepackage[final]{epsfig}
\usepackage[left=3cm,right=3cm, top=2.5cm,bottom=2.5cm,bindingoffset=0cm]{geometry}
\usepackage{graphics}
\usepackage{amsmath}
\usepackage{amsfonts}
\usepackage{latexsym}
\usepackage{amssymb, mathabx}
\usepackage{amsthm}
\usepackage{graphicx}
\usepackage{epstopdf}
\usepackage{multicol,multirow}
\usepackage{hyperref, enumitem}
\usepackage{mathtools}
\usepackage{setspace}

\usepackage{mathrsfs}

 \usepackage{relsize}

\usepackage[nottoc]{tocbibind}

\usepackage{amsthm}

\usepackage[latin1]{inputenc}
\usepackage{tikz}
\usetikzlibrary{shapes,arrows,decorations.pathmorphing}
\usepackage{float}

\usepackage{blkarray}
\usepackage{xypic}

\usepackage[title]{appendix}

\makeatletter
\def\thmhead@plain#1#2#3{%
  \thmname{#1}\thmnumber{\@ifnotempty{#1}{ }\@upn{#2}}%
  \thmnote{ {\the\thm@notefont#3}}}
\let\thmhead\thmhead@plain

\makeatother

\newcounter{AppCounter}

\def\restrict#1{\raise-.5ex\hbox{\ensuremath|}_{#1}}

\newtheorem{lemma}{Lemma}[section]
\newtheorem{proposition}[lemma]{Proposition}

\newtheorem{remark-definition}[lemma]{Remark-Definition}
\newtheorem{theorem}[lemma]{Theorem}
\newtheorem{corollary}[lemma]{Corollary}

\newtheorem{proposition-conjecture}[lemma]{Proposition-conjecture}

\newtheorem{question}[lemma]{Question}

\theoremstyle{definition}

\newtheorem{definition}[lemma]{Definition}
\newtheorem{remark}[lemma]{Remark}

\newcommand{\marginnote}[1]
{
}

\newcounter{cy}

\newcounter{bk}

\newcounter{dps}

\usepackage{lipsum}

\title{Vortex motion of the Euler and Lake equations}

\author{Cheng Yang\thanks{Department of Mathematics, University of Toronto, Toronto, ON M5S 2E4, Canada, and the Fields Institute, Toronto, ON M5T 3J1, Canada; 
 e-mail: \tt{chyang@math.toronto.edu}
  }}

\begin{document}

\doublespacing
\date{}
\maketitle

\begin{abstract}
We start by surveying the planar point vortex motion of the Euler equations in the whole plane, half-plane and quadrant. Then we go on to prove the non-collision property of the 2-vortex system by using the explicit form of orbits of the 2-vortex system in the half-plane. We also prove that the $N$-vortex system in the half-plane is nonintegrable for $N>2$, which was suggested previously by numerical experiments without rigorous proof.

The skew-mean-curvature  (or binormal) flow in $\mathbb R^n,\;n\geq3$ with  certain symmetry can be regarded as point vortex motion of the 2D lake equations. We compare point vortex motions of the Euler and lake equations. Interesting similarities between the point vortex motion in the half-plane, quadrant and the binormal motion of coaxial vortex rings, sphere product membranes are addressed. We also raise some open questions in the paper. 

\end{abstract}

\tableofcontents

\newpage
\section{Introduction} \label{intro}

In his seminal work \cite{Helm}, Helmholtz built the foundation for the theory of vortex dynamics. In the same paper, he initiated the study of two areas: interacting point vortex motion and the motion of vortex rings. After Helmholtz's work, these two areas develop independently along different lines. There are many important studies of point vortex motions, for instance, the works of Kirchoff \cite{Kirch},  Gr\"{o}bli \cite{Grob}, Synge \cite{Syn}, Aref \cite{Aref}, Ziglin \cite{Zig}, etc. Meanwhile, the motion of vortex rings has become a classical subject of binormal flow of vortex filaments \cite{DaRi} and is studied by W. Thomson \cite{Thom}, J. J. Thomson \cite{ThomJ}, Dyson \cite{Dyso}, Hicks \cite{Hick}, Saffman \cite{Saff}, etc. In this paper, we search for possible relations between the motions  of point vortices and vortex rings. We will survey related research and prove some new results.

First we survey the vortex motion of the Euler equations. We start with the planar point vortex system and consider three cases: the whole plane, half-plane and quadrant. It is well-known that on the whole plane the $N$-vortex system for $N\leq 3$ is integrable, and the explicit solution of $3$-vortex system is given by Gr\"{o}bli \cite{Grob}. For $N>3$, Ziglin \cite{Zig} proves that vortex systems with more than 3 vortices are nonintegrable.

We study the integrability and nonintegrability properties of point vortex systems in the half-plane and quadrant. By symmetry group consideration, we know that the $N$-vortex system in the half-plane is integrable when $N\leq 2$ . In Section \ref{sec:planar_vortex}, we give the explicit solutions for 2-vortex systems in the half-plane, which we did not find in the previous literature. Furthermore, using the explicit form of the solutions, we show a non-collision property of the 2-vortex system in the half-plane.

 There are  some numerical evidences showing the nonintegrability of $N$-vortex system in the half-plane for $N>2$, see cf. \cite{KnCo}, but as far as we know the rigorous proof is missing. So in this paper we present a proof. Curiously, up to the order of $\epsilon$, the perturbed system we used to prove the nonintegrability of 3 point vortices in the half-plane appears to be the same as the perturbed system used in \cite{BaBa} to prove the nonintegrability of 3 coaxial vortex rings. This observation shows an interesting relation of point vortex system in the half-plane and coaxial vortex ring system, this similarity also is indicated by some phenomena, for example both 2 point vortices and 2 vortex rings can have leapfrogging motion. (See Figure 2 in \cite{Lim} for an illustration of the Leapfrogging motion.) Explicit solution of a single vortex in the quadrant is also given in Section 2.

One of the purposes of this paper is to compare the point vortex system in the half-plane with the system of coaxial vortex rings, as well as, compare single vortex motion in the quadrant with motion of sphere product vortex membrane under the binormal flow. 
In Section \ref{sec:rings}, we describe Dyson's model of coaxial vortex rings system. In Section \ref{sec:SphProd}, we study the motion of sphere product vortex membrane under the skew-mean-curvature flow. 

Under the axisymmetry and sphere product symmetry, the Euler equations can be reduced to the lake equations. Hence the two cases studied in Section \ref{sec:rings} and \ref{sec:SphProd} can be seen as the point vortex motion for the lake equations. So in this paper we are actually comparing the point vortex systems of the incompressible Euler equations and the lake equations. Several open problems are raised in the paper.

\medskip

{\bf Acknowledgments.} The author is grateful to Boris Khesin and Klas Modin for stimulating discussions.


\medskip


\section{Planar point vortex system}\label{sec:planar_vortex}
\subsection{Point vortex system in the plane}\label{subsec:plane}

The dynamics of $N$ vortices with strength $\Gamma_i$ located at ${\bf x_i}=(x_i(t),y_i(t))$, $i=1,2,\dots,N$
are governed by the equations:
\begin{equation}\label{eq:plane}
\left\{\begin{array}{l}
\dot{x_i}=-\frac{1}{2\pi}\sum_{j\neq i}^N\frac{\Gamma_j(y_i-y_j)}{l_{ij}^2}, \\\\
\dot{y_i}=\;\;\frac{1}{2\pi}\sum_{j\neq i}^N\frac{\Gamma_j(x_i-x_j)}{l_{ij}^2},
\end{array}\right.
\end{equation}
where $l_{ij}=\sqrt{(x_i-x_j)^2+(y_i-y_j)^2}$ are the distances between point vortices.

In the space $(\mathbb R^2)^N$ equipped with the symplectic structure $\omega=\sum_{i=1}^N \Gamma_i dx_i\wedge dy_i$, (\ref{eq:plane}) can be written in the Hamiltonian form:
\begin{equation}\label{eq:plane_Ham}
\left\{\begin{array}{l}
\Gamma_i\dot{x_i}=\;\;\frac{\partial H}{\partial y_i}, \\\\
\Gamma_i\dot{y_i}=-\frac{\partial H}{\partial x_i},
\end{array}\right.
\end{equation}
whose Hamiltonian function is $H({\bf x_1},{\bf x_2},\dots, {\bf x_N})=-\frac {1}{4\pi}\sum_{j\neq i}\Gamma_i\Gamma_j\log(l_{ij})$. 

Besides the vortex interaction energy given by the Hamiltonian function $H$, the following 3 quantities are also conserved under the Hamiltonian flow of (\ref{eq:plane_Ham}):
\begin{equation}
Q=\sum_{i=1}^N\Gamma_ix_i,\;P=\sum_{i=1}^N\Gamma_iy_i\;\text{and}\;I=\sum_{i=1}^N\Gamma_i(x_i^2+y_i^2).
\end{equation}

\begin{remark}
The Hamiltonian system (\ref{eq:plane_Ham}) is invariant under the action of semidirect product group $\text{SO(3)}\ltimes \mathbb R^2$, the above invariants $Q,\,P,\,I$ can be regarded as the momentum map of this group action. Applying the symplectic reduction theorem (see cf. \cite{Mars}), if the dimension of the reduced system is no more than 2, then the original system (\ref{eq:plane_Ham}) is integrable.

For the point vortex systems on the sphere or hyperbolic plane, the corresponding symmetry group are \text{SO(3)} and \text{SL(2)} respectively, then one can obtain similar integrable results by reduction theorem, see \cite{MoVi} for more details. 
\end{remark}

Hence, we have the following theorem.
\begin{theorem}\label{thm:integ}{\rm (cf. \cite{Newt})} The N-vortex problem for $N\leq 3$ is integrable. If the total strength of the vortices $\Gamma=\sum_{i=1}^N\Gamma_i=0$, the 4-vortex problem is integrable.
\end{theorem}
\begin{proof}
The quantities $H,\;I,\;P^2+Q^2$ are mutually involution and functional independent. 

If $\Gamma=0$, then $ H,\;P,\;Q,\;I$  are mutually involution and functional independent. 
\end{proof}

\begin{remark}
When $N=2$ and $\Gamma_1+\Gamma_2\neq 0$, the distance $l_{12}$ and the center of vorticity $C=\frac{\Gamma_1{\bf x_1}+\Gamma_2{\bf x_2}}{\Gamma_1+\Gamma_2}$ are conserved. Therefore, the two vortices rotate over concentric circles about their center of vorticity. When $\Gamma_1+\Gamma_2=0$ (the dipole case), the vortices move along the perpendicular bisector with velocity $(1/2\,(\Gamma_1^2+\Gamma_2^2))^{1/2}/(2\pi l_{12})$.
\end{remark}

\begin{remark}
It is useful to rewrite the $N$-vortex system  (\ref{eq:plane}) in terms of the vortex separations $l_{ij}$. When $N=3$, explicit solutions of the aforementioned system of vortex separations $l_{ij}$ is studied first by Gr\"{o}bli \cite{Grob}, and later reconsidered by Novikov \cite{Nov} and Aref \cite{Aref}.
\end{remark}

\begin{remark}
When $N\geq 4$, the system (\ref{eq:plane_Ham}) is nonintegrable in general, see cf. \cite{Zig}.
\end{remark}

Finally, for a three-vortex system satisfying $\Gamma_1\Gamma_2+\Gamma_2\Gamma_3+\Gamma_3\Gamma_1= 0$ and $\Gamma_1\Gamma_2l_{12}^2+\Gamma_2\Gamma_3l_{23}^2+\Gamma_3\Gamma_1l_{31}^2= 0$, there exist solutions such that the triangle of vortices collapses self-similarly to a point in finite time. {\bf Self-similar vortex collapse}  may also be found analytically for four and five vortices \cite{Onei}.

\begin{proposition}\label{cond_collapse}
The condition $\sum\limits_{1\leq i<j\leq N}\Gamma_i\Gamma_j=0$ is necessary for self-similar vortex collapse.
\end{proposition}
\begin{proof}
For a self-similar motion, one can find a function $f(t)$ of time $t$, such that $l_{ij}(t)=f(t)l_{ij}(0)$, where $f(0)=1$ and $\lim\limits_{t\rightarrow\infty}f(t)=0$. Therefore the Hamiltonian at time $t$ is 
$$
H(t)=\frac {1}{4\pi}\sum_{j\neq i}\Gamma_i\Gamma_j\log(f(t)l_{ij}(0))=\frac {1}{4\pi}\sum_{j\neq i}\Gamma_i\Gamma_j\log(f(t))+\frac {1}{4\pi}\sum_{j\neq i}\Gamma_i\Gamma_j\log(l_{ij}(0)).
$$
Because $H(t)$ is conserved, the first term $\frac {1}{4\pi}\sum_{j\neq i}^N\Gamma_i\Gamma_j\log(f(t))$ must be 0 for all time $t$, which  gives us $\sum\limits_{1\leq i<j\leq N}\Gamma_i\Gamma_j=0$.
\end{proof}

\begin{remark}
From this proof, one can see that the condition $\sum\limits_{1\leq i<j\leq N}\Gamma_i\Gamma_j=0$ appeared in \cite{Aref2} is actually a necessary condition for self-similar motion. Therefore $\sum_{1\leq i<j\leq N}\Gamma_i\Gamma_j\neq0$ does not rule out other possibilities of vortex collision which are not self-similar vortex collapse.

For an $N$-vortex system in the half-plane, if we count both the vortices in the half-plane and their images, one can see that $\sum_{1\leq i<j\leq 2N}\Gamma_i\Gamma_j<0$ as a $2N$-system in the whole plane, hence by Proposition \ref{cond_collapse},  self-similar vortex collapse can not happen in the half-plane, but there still could be other kinds of collision.
\end{remark}

\medskip

\subsection{Point vortex system in the half-plane}\label{subsec:half_plane}

Before discussing the vortex motion in the half-plane, let us first consider general  point vortex system in a domain $D\subset\mathbb R^2$. The Green function $G_D(x,x'): D\times D\rightarrow \mathbb R$ associated to the domain $D$ is
$$
G_D(x,x')=G(x,x')+\gamma_D(x,x'),
$$
where $G(x,x')$ is the Green function in the whole plane $\mathbb R^2$, and smooth function $\gamma_D(x,x')$ is symmetric: $\gamma_D(x,x')=\gamma_D(x',x)$.

Now suppose that $N$ vortices with strength $\Gamma_i$ locate at ${\bf x_i}(t)=(x_i(t),y_i(t))$, $i=1,2,\dots,N$ in the domain $D$. Then the system of vortices in $D$ satisfies
\begin{equation}\label{eq:domain}
\frac {d}{dt}{\bf x_i}(t)=\nabla_i^{\perp}\sum_{j\neq i}^N\Gamma_jG_D({\bf x_i}(t),{\bf x_j}(t))+\frac 12 \Gamma_i\nabla_i^{\perp}\hat{\gamma_D}({\bf x_i}(t)),
\end{equation}
where $\hat{\gamma_i}({\bf x_i}(t))=\gamma_i({\bf x_i}(t),{\bf x_i}(t))$ and $\nabla_i^{\perp}$ stands for $(\partial_{y_i},-\partial_{x_i})$. See \cite{MaPu} for more details.

\begin{remark}
Unlike the vortex system in the plane, the vortex system in a general domain contains a self-interaction part (term $\frac 12 \Gamma_i\nabla_i^{\perp}\gamma_D({\bf x_i}(t))$ in the equation).
\end{remark}

This system is also Hamiltonian \cite{Lin} (see (\ref{eq:plane_Ham})):
\begin{equation}\label{eq:domain_Ham}
\left\{\begin{array}{l}
\Gamma_i\dot{x_i}=\;\;\frac{\partial H}{\partial y_i}, \\\\
\Gamma_i\dot{y_i}=-\frac{\partial H}{\partial x_i},
\end{array}\right.
\end{equation}
with Hamiltonian function 
\begin{equation}\label{Ham}
H({\bf x_1},{\bf x_2},\dots, {\bf x_N})=\frac 12\sum_{j\neq i}^N\Gamma_i\Gamma_jG_D({\bf x_i},{\bf x_j})+\frac 12 \sum_{i=1}^N\Gamma_i^2\hat{\gamma_D}({\bf x_i}).
\end{equation}

Next we apply the above formulas to the half-plane case: $D=\mathbb R^2_+=\{(x,\,y):\,y\geq0\}$, the corresponding Green function is 
\begin{equation}
G_{\mathbb R^2_+}({\bf x},{\bf x'})=-\frac {1}{2\pi}\log\|{\bf x}-{\bf x'}\|+\frac {1}{2\pi}\log\|{\bf x}-{\bf x'^*}\|,
\end{equation}
where ${\bf x'^*}=(x',-y')$  stands for the mirror image of ${\bf x'}=(x',y')$. Note that the first term of $G_{\mathbb R^2_+}({\bf x},{\bf x'})$ is the Green function in $\mathbb R^2$, hence the second term is responsible for the self-interaction:
$$
\hat{\gamma_{\mathbb R^2_+}}({\bf x})=\gamma_{\mathbb R^2_+}({\bf x},{\bf x})=\frac{1}{2\pi}\log(2y),
$$
where ${\bf x}=(x,y)$ with $y\geq 0$ is a point in the half-plane.

Put the expression of $G_{\mathbb R^2_+}$ and $\hat{\gamma_{\mathbb R^2_+}}$ into (\ref{Ham}), one obtains the Hamiltonian function for the $N$-vortex system in the half-plane $\mathbb R^2_+$:
\begin{equation}\label{Ham_half}
H_{\mathbb R^2_+}({\bf x_1},{\bf x_2},\dots, {\bf x_N})=\frac {1}{4\pi}\sum_{j\neq i}\Gamma_i\Gamma_j\log\frac{(x_i-x_j)^2+(y_i+y_j)^2}{(x_i-x_j)^2+(y_i-y_j)^2}+\frac {1}{2\pi} \sum_{i=1}^N\Gamma_i^2\log(2y_i).
\end{equation}

\begin{remark}
When $N=1$, the Hamiltonian (\ref{Ham_half}) becomes $H_{\mathbb R^2_+}(x_1,y_1)=\frac{1}{2\pi}\Gamma_1^2\log(2y_1)$ where $(x_1,y_1)$ is the position of  a single vortex and $\Gamma_1$ is its strength. This system can be solved explicitly, the single vortex moves with a speed  inversely proportional to the distance from the  $x$-axis in a straight line parallel to the $x$-axis.
\end{remark}

For the motion of two point vortices in the half-plane, the explicit solutions of Hamiltonian (\refeq{Ham_half}) are presented in Sections \ref{subsec:2-vortex1} and \ref{subsec:2-vortex2}. Suppose that in the half-plane $\mathbb R^2_+$, 2 vortices  located at ${\bf x_1}=(x_1(t),y_1(t))$ and ${\bf x_2}=(x_2(t),y_2(t))$ have strengths $\Gamma_1$ and $\Gamma_2$ respectively. The symplectic structure is $\omega=\Gamma_1dx_1\wedge dy_1+\Gamma_2dx_2\wedge dy_2$, and the corresponding Hamiltonian is 
\begin{equation}\label{Ham_half_2}
\begin{array}{rcl}
H_{\mathbb R^2_+}(x_1,x_2,y_1,y_2)&=&\frac {1}{2\pi}\{\Gamma_1\Gamma_2\log[(x_1-x_2)^2+(y_1+y_2)^2]\\
& &-\Gamma_1\Gamma_2\log[(x_1-x_2)^2+(y_1-y_2)^2]\\
& &+\Gamma_1^2\log (2y_1)+\Gamma_2^2\log (2y_2)\}\\
&=&\frac {1}{2\pi}\log\left\{(2y_1)^{\Gamma_1^2}(2y_2)^{\Gamma_2^2}\left[\frac{(x_1-x_2)^2+(y_1+y_2)^2}{(x_1-x_2)^2+(y_1-y_2)^2}\right]^{\Gamma_1\Gamma_2}\right\}.
\end{array}
\end{equation}


\medskip

\subsubsection{The motion of two generic point vortices}\label{subsec:2-vortex1}

First we consider the generic case when $\Gamma_1+\Gamma_2\neq 0$. Let us introduce the notations for  the center of vorticity:
$$
x_0=\frac{\Gamma_1x_1+\Gamma_2x_2}{\Gamma_1+\Gamma_2},\;\;y_0=\frac{\Gamma_1y_1+\Gamma_2y_2}{\Gamma_1+\Gamma_2},
$$
and relative coordinates:
$$
x_r=x_1-x_2,\;\;y_r=y_1-y_2.
$$

We know that $y_0$ and the Hamiltonian function are first integrals of the system. (Here $y_0$ can be regarded as the momentum map of the Abelian group $\mathbb R$-action on $\mathbb R^2_{+}$.) We can fix a value of the momentum map $y_0=\mu$ and an energy level $H=E$, and describe the orbits of the vortices in the 2-dimensional reduced space $(x_r,y_r)$.
\begin{theorem}\label{thm:2-vortex}
For 2 vortices with strengths $\Gamma_1$ and $\Gamma_2$ such that $\Gamma_1+\Gamma_2\not=0$, located at ${\bf x_1}=(x_1(t),y_1(t))$ and ${\bf x_2}=(x_2(t),y_2(t))$ respectively in the half-plane, introduce parameters $y_0=\mu$ and $H=E$. Then the orbits in the reduced space $(x_r,y_r)$  satisfy the following equation:
\begin{equation}\label{eq:2-vortex}
\left(\mu+\frac{\Gamma_2}{\Gamma_1+\Gamma_2}y_r\right)^{\Gamma_1^2}\left(\mu-\frac{\Gamma_1}{\Gamma_1+\Gamma_2}y_r\right)^{\Gamma_2^2}\left[\frac{x_r^2+\left(2\mu+\frac{\Gamma_2-\Gamma_1}{\Gamma_1+\Gamma_2}y_r\right)^2}{x_r^2+y_r^2}\right]^{\Gamma_1\Gamma_2}=e^{2\pi E}
\end{equation}
\end{theorem}
\begin{proof}
First take a canonical transformation 
$$
(x_1,y_1,x_2,y_2)\mapsto\left(x_0,(\Gamma_1+\Gamma_2)y_0,x_r,\frac{\Gamma_1\Gamma_2}{\Gamma_1+\Gamma_2}y_r\right),
$$
in the new coordinates the Hamiltonian becomes
\begin{equation}\label{eq:Ham_half3}
\begin{array}{rcl}
& &H_{\mathbb R^2_+}\left(x_0,(\Gamma_1+\Gamma_2)y_0,x_r,\frac{\Gamma_1\Gamma_2}{\Gamma_1+\Gamma_2}y_r\right)\\
&=&\frac{1}{2\pi}\log\left\{\left(y_0+\frac{\Gamma_2}{\Gamma_1+\Gamma_2}y_r\right)^{\Gamma_1^2}\left(y_0-\frac{\Gamma_1}{\Gamma_1+\Gamma_2}y_r\right)^{\Gamma_2^2}\left[\frac{x_r^2+\left(2y_0+\frac{\Gamma_2-\Gamma_1}{\Gamma_1+\Gamma_2}y_r\right)^2}{x_r^2+y_r^2}\right]^{\Gamma_1\Gamma_2}\right\}.
\end{array}
\end{equation}
One can see that this new Hamiltonian does not depend on $x_0$, hence $y_0$ is a conserved quantity. By fixing $y_0=\mu$ and the Hamiltonian $H_{\mathbb R^2_+}=E$, we obtain the equation for orbits in $(x_r,y_r)$ coordinates:
$$
\left(\mu+\frac{\Gamma_2}{\Gamma_1+\Gamma_2}y_r\right)^{\Gamma_1^2}\left(\mu-\frac{\Gamma_1}{\Gamma_1+\Gamma_2}y_r\right)^{\Gamma_2^2}\left[\frac{x_r^2+\left(2\mu+\frac{\Gamma_2-\Gamma_1}{\Gamma_1+\Gamma_2}y_r\right)^2}{x_r^2+y_r^2}\right]^{\Gamma_1\Gamma_2}=e^{2\pi E}.
$$
\end{proof}
A corollary of this theorem is the following non-collision property.
\begin{corollary}\label{thm:2-vortex_collision1}
If $\Gamma_1+\Gamma_2\neq 0$, the two point vortices in the half-plane will not collide or hit the boundary ($x$-axis).
\end{corollary}
\begin{proof}
There are 3 possibilities of collision, and we exclude them one by one.

1. The 2 vortices hit each other, which means $x_r=y_r=0$ in finite time. We claim that if $x_r=y_r=0$ at a certain time, the Hamiltonian function (\ref{eq:Ham_half3}) becomes infinity.  First if $y_0\neq 0$, it is clear that the function (\ref{eq:Ham_half3}) becomes infinity as $x_r=y_r=0$. And if $y_0=0$, the rational function $\left(y_0+\frac{\Gamma_2}{\Gamma_1+\Gamma_2}y_r\right)^{\Gamma_1^2}\left(y_0-\frac{\Gamma_1}{\Gamma_1+\Gamma_2}y_r\right)^{\Gamma_2^2}\left[\frac{x_r^2+\left(2y_0+\frac{\Gamma_2-\Gamma_1}{\Gamma_1+\Gamma_2}y_r\right)^2}{x_r^2+y_r^2}\right]^{\Gamma_1\Gamma_2}=0$, so after taking the logarithm the Hamiltonian (\ref{eq:Ham_half3}) becomes infinity. However the Hamiltonian should not  become infinity because it is conserved, hence 2 vortices will not hit each other. 

2. One vortex hits the boundary while the other does not. In this case using the expression of Hamiltonian function (\ref{Ham_half_2}), one can see that if $y_1=0, y_2\neq 0$, the Hamiltonian becomes infinity. Again this is a contradiction
since the Hamiltonian is a conserved quantity.

3. Both vortices hit the boundary. If $y_0\neq 0$ at the initial time, then this case is impossible because the center of vorticity $y_0$ is conserved for all time. And if $y_0=0$, put $y_r=0$ and $x_r\neq 0$ (we only consider $x_r\neq 0$ here since $x_r=0$ is already studied in case 1) into  Hamiltonian (\ref{eq:Ham_half3}), we see that the Hamiltonian becomes infinity which is a contradiction.
\end{proof}

\begin{remark}
By changing the parameters $\mu$ and $E$, we can study the motion of 2 vortices. For example if the 2 vortices have the same strength, the equation in terms of relative position is (modulo certain constants)
$$
\frac{1}{1-y_r^2}-\frac{1}{1+x_r^2}=\exp(-E).
$$
\end{remark}


\medskip

\subsubsection{The motion of a vortex dipole}\label{subsec:2-vortex2}

Next let us consider the motion of the 2-vortex system in the half-plane when $\Gamma_1=-\Gamma_2$. Suppose that $\Gamma_1=-\Gamma_2=1$ and introduce
the new (center and relative) coordinates:
$$
x_0=\frac{x_1+x_2}{2},\;\;y_0=\frac{y_1+y_2}{2},
$$
and
$$
x_r=x_1-x_2,\;\;y_r=y_1-y_2.
$$

In this case, $y_r$ and Hamiltonian function are first integrals of the system. Hence we fix $y_r=\nu$ and $H=E$, and describe the orbits of the vortices in the 2-dimensional reduced space $(x_r,y_0)$.
\begin{theorem}\label{thm:2-vortex-opposite}
For the dipole case, i.e. 2 vortices with strengths $\Gamma_1=1$ and $\Gamma_2=-1$ located at ${\bf x_1}=(x_1(t),y_1(t))$ and ${\bf x_2}=(x_2(t),y_2(t))$ respectively in the half-plane, set parameters $y_r=\nu$ and $H=E$. Then the orbits in the reduced space $(x_r,y_0)$ satisfy the following equation:
\begin{equation}\label{eq:2-vortex-opposite}
\frac{1}{\nu^2+x_r^2}+\frac{1}{4y_0^2-\nu^2}=e^{-2\pi E}
\end{equation}
\end{theorem}
\begin{proof}
One can check that $dx_1\wedge dy_1-dx_2\wedge dy_2=dx_0\wedge dy_r+dx_r\wedge dy_0$, so the map $(x_1,y_1,x_2,y_2)\mapsto(x_0,y_r,x_r,y_0)$ is a canonical transformation,
and in the new coordinates, the Hamiltonian becomes
$$
H_{\mathbb R^2_+}\left(x_0,y_r,x_r,y_0\right)=\frac{1}{2\pi}\log\left\{\frac{(4y_0^2-y_r^2)(x_r^2+y_r^2)}{x_r^2+4y_0^2}\right\}.
$$ 
This new Hamiltonian does not depend on $x_0$, hence $x_0$'s conjugate coordinate $y_r$ is a conserved quantity. By fixing $y_r=\nu$ and $H_{\mathbb R^2_+}=E$, we obtain the equation for orbits in $(x_r,y_0)$ coordinates:
$$
\frac{1}{\nu^2+x_r^2}+\frac{1}{4y_0^2-\nu^2}=e^{-2\pi E}.
$$
\end{proof}
\begin{remark}
Here $\nu=y_r=y_1-y_2$, so $\nu=0$ means that the vortex dipole is symmetric with respect to their vertical bisector, and this case can be reduced to a single vortex in the quadrant. (See Section \ref{subsec:quadrant}.)

In general, an $N$-vortex system in the half-plane gives rise to a $2N$-vortex system in the plane ($N$ vortices plus $N$ images). Similarly, an $N$-vortex system in the quadrant gives rise to a $2N$-vortex system in the half plane, which also can be seen as a $4N$-vortex system.
\end{remark}

\begin{remark}
To study Equation (\ref{eq:2-vortex-opposite}), for simplicity, we assume that the RHS of (\ref{eq:2-vortex-opposite}) $e^{-2\pi E}=C>0$. Suppose that $x_r=0$ then Equation (\ref{eq:2-vortex-opposite}) becomes
$$
\frac{1}{\nu^2}+\frac{1}{4y_0^2-\nu^2}=C.
$$ 
Then solve this equation for $y_0$, we get $4y_0^2=\frac{C\nu^2}{C-1/\nu^2}$. Hence when $C>\frac{1}{\nu^2}$, $y_0$ has real solutions, i.e., $x_r$ can be 0, otherwise, if $C\leq\frac{1}{\nu^2}$, $x_r$ can not be 0.

Note that here we have two parameters: $C=e^{-2\pi E}>0$ is related to the energy and $\nu=y_1-y_2$ is the difference of 2 vortices in $y$-direction.

From the above discussion, we know that when $C>\frac{1}{\nu^2}$, a typical orbit in $(x_r,y_0)$ coordinates looks like the following graph.
\begin{figure}[H]
\centering
\includegraphics[width = 0.7\textwidth]{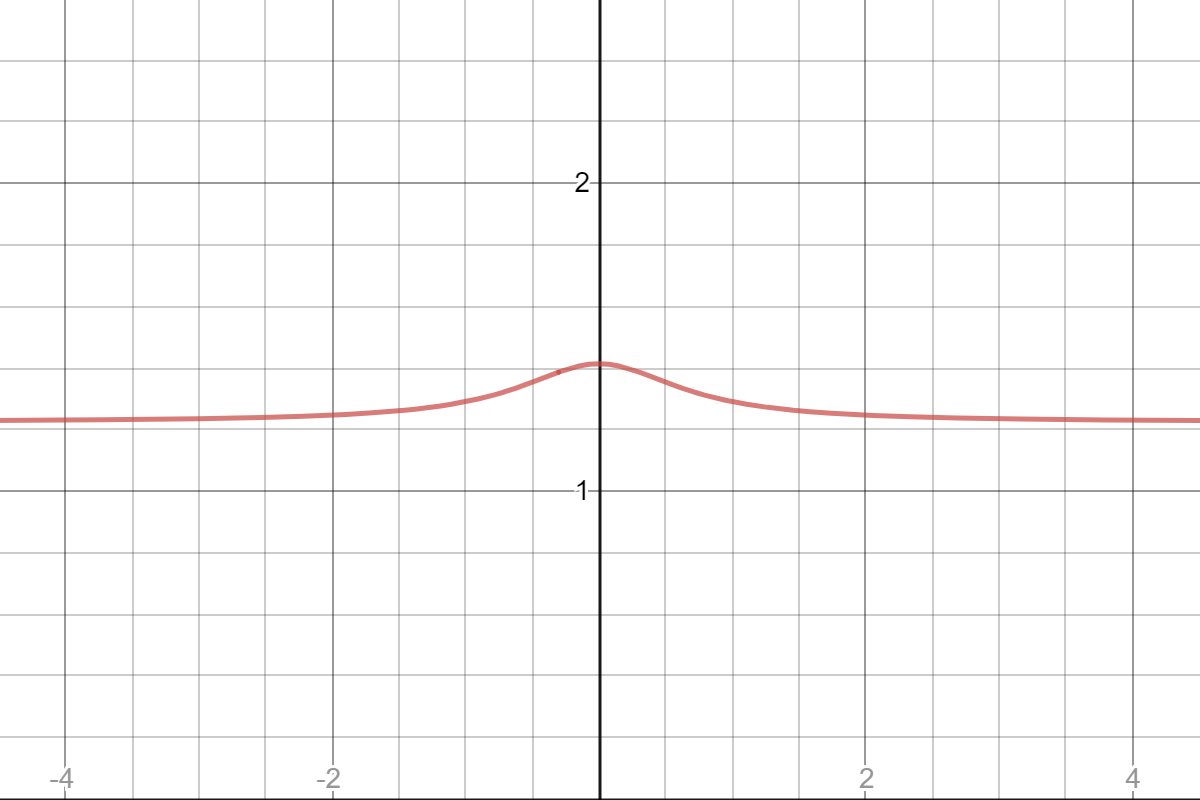}
\caption{Orbit when $C>\frac{1}{\nu^2}$}
\end{figure}

And when $0<C\leq\frac{1}{\nu^2}$, an orbit in $(x_r,y_0)$ variables looks like the following graph.
\begin{figure}[H]
\centering
\includegraphics[width = 0.7\textwidth]{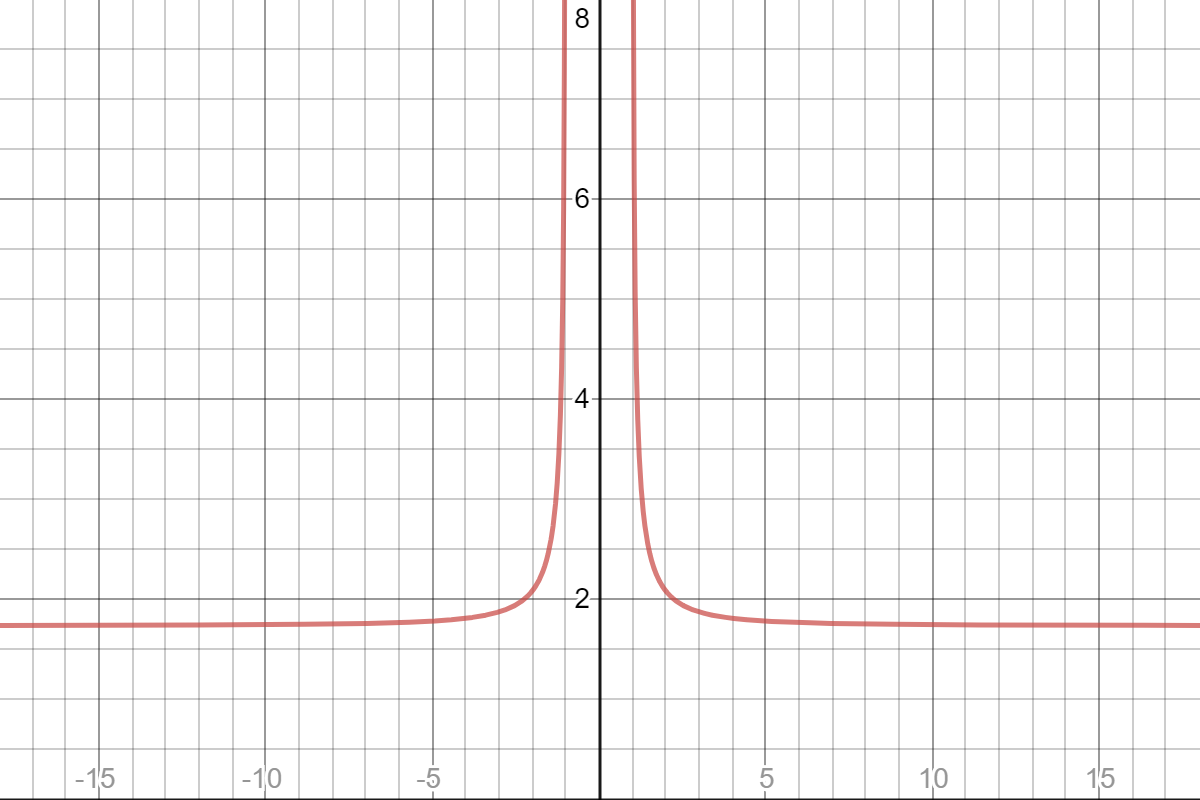}
\caption{Orbit when $0<C\leq\frac{1}{\nu^2}$}
\end{figure}
\end{remark}
\begin{remark}
For orbits  when $C\leq\frac{1}{\nu^2}$, we can study the asymptotes. If $y_0\to \pm\infty$, we have $\frac{1}{x_r^2+\nu^2}\to C$, i.e., $|x_r|\to\sqrt{1/C-\nu^2}$, so $x_r^2+y_r^2\to1/C$. This means if two vortices move up together, the distance between them approaches $1/C$ (slant asymptotes in the regular $(x_1,x_2,y_1,y_2)$ coordinates). 

If $x_r\to \pm\infty$, we have $y_0\to \frac 12\sqrt{1/C+\nu^2}$. This means if two vortices move apart each other in $x$-direction, their center in $y$-direction approaches  $\frac12\sqrt{1/C+\nu^2}$ (horizontal asymptotes).
\end{remark}
Similar to the previous case, one can prove the following property on collision.
\begin{corollary}\label{thm:2-vortex_collision2}
The vortex dipole in the half-plane will not collide or hit the boundary ($x$-axis).
\end{corollary}
\begin{proof}
Again, we just need to check the three cases in the proof of Corollary \ref{thm:2-vortex_collision1}. For Case 1 and 3, we have $y_r=0$. Since for vortex dipole, $y_r$ is a conserved quantity, we know that the two vortices have the same $y$ coordinate for all time $t$, hence one can reduce the motion of vortex dipole to a single vortex in the quarter-plane, and we can see in the later section, a single vortex moves in the quarter-plane will not hit the boundary, therefore Case 1 and 3 will not happen.

As for Case 2, in the expression of Hamiltonian function (\ref{Ham_half_2}), one can take for instance $y_1=0,\;y_2\neq 0$ and $\Gamma_1=1, \Gamma_2=-1$, and the Hamiltonian becomes infinity which is a contradiction.
\end{proof}

Corollary \ref{thm:2-vortex_collision1} and \ref{thm:2-vortex_collision2} together show that $N$-vortex system in the half-plane will not collide or hit the boundary in finite time when $N=2$, and one would ask for the case when $N>2$.

\begin{question}\label{q2}
Is it possible that an $N$-vortex system in the half-plane collides or hits the boundary in  finite time for $N>2$? 
\end{question}

It is easy to see that for $N=2$ the system is integrable, and we presented its explicit solutions above. So a natural question to ask is 
\begin{question}\label{q3}
Prove nonintegrability for the $N$-vortex system in the half-plane $\mathbb R^2_+$ when $N>2$.
\end{question}

We answer the second question in the following theorem.

\begin{theorem}\label{nonint_half}
The $N$-vortex system in the half-plane $\mathbb R^2_+$ when $N>2$ is nonintegrable.
\end{theorem}
\begin{proof}
First we consider a restricted 3-vortex system with vortices whose strengths are $\Gamma_1=\Gamma_2=1$, and $\Gamma_3=0$, thus the Hamiltonian for the motion of vortices with strengths $\Gamma_1$ and $\Gamma_2$ is 
$$
H_{12}(x_1,y_1,x_2,y_2)=\frac{1}{2\pi}\log \left(4y_1y_2\frac{(x_1-x_2)^2+(y_1+y_2)^2}{(x_1-x_2)^2+(y_1-y_2)^2}\right),
$$
and the Hamiltonian for the 0-strength vortex is 
$$
H_3(x_3,y_3,t)= \frac{1}{2\pi}\log\frac{[(x_1-x_3)^2+(y_1-y_3)^2][(x_2-x_3)^2+(y_2-y_3)^2]}{[(x_1-x_3)^2+(y_1+y_3)^2][(x_2-x_3)^2+(y_2+y_3)^2]}.
$$
Assume that $x_i=\epsilon\tilde{x_i}$, $y_i=1+\epsilon\tilde y_i$, $i=1,\,2,\,3$, i.e., vortices are close to each other and relatively far from the boundary. Then $H_{12}$ becomes
\begin{equation}\label{eq:H_12}
\begin{array}{rcl}
&&\tilde H_{12}(\tilde{x_1},\tilde{y_1},\tilde{x_2},\tilde{y_2})=\frac{1}{2\pi}\log\left[4(1+\epsilon\tilde y_1)(1+\epsilon\tilde y_2)\frac{\epsilon^2(\tilde x_1-\tilde x_2)^2+(2+\epsilon(\tilde y_1+\tilde y_2))^2}{\epsilon^2(\tilde x_1-\tilde x_2)^2+\epsilon^2(\tilde y_1-\tilde y_2)^2}\right]\\\\
=&&\frac{1}{2\pi}\log\left[\frac{16}{\epsilon^2}\left(\frac{1}{(\tilde x_1-\tilde x_2)^2+(\tilde y_1-\tilde y_2)^2}+\epsilon\frac{2(\tilde y_1+\tilde y_2)}{(\tilde x_1-\tilde x_2)^2+(\tilde y_1-\tilde y_2)^2}+O(\epsilon^2)\right)\right]\\\\
=&&\frac{1}{2\pi}\log\frac{16}{\epsilon^2}+\frac{1}{2\pi}\log\frac{1}{(\tilde x_1-\tilde x_2)^2+(\tilde y_1-\tilde y_2)^2}+\frac{\epsilon}{\pi}(\tilde y_1+\tilde y_2)+O(\epsilon^2).

\end{array}
\end{equation}
In the above Hamiltonian function, the leading term is $-\frac{1}{2\pi}\log[(\tilde x_1-\tilde x_2)^2+(\tilde y_1-\tilde y_2)^2]$ (ignoring the constant term $\frac{1}{2\pi}\log\frac{16}{\epsilon^2}$), which means, to the leading order, the two vortices with strength $\Gamma_1$ and $\Gamma_2$ rotate on  a common circle. Let $\xi_1=(\tilde x_1-\tilde x_2)/2+i\tilde y_1$ and  $\xi_2=(\tilde x_2-\tilde x_1)/2+i\tilde y_2$ and suppose that initially $\xi_1(0)=i$ and $\xi_2(0)=-i$, then $\xi_1(t)=ie^{i\omega_0t}+O(\epsilon)$ and $\xi_2(t)=-ie^{i\omega_0t}+O(\epsilon)$, hence we have
$$
\tilde x_1(t)=-\sin(\omega_0t)+O(\epsilon),\;\tilde x_2(t)=\sin(\omega_0t)+O(\epsilon), 
$$
$$
\tilde y_1(t)=\cos(\omega_0t)+O(\epsilon),\;\tilde y_1(t)=-\cos(\omega_0t)+O(\epsilon).
$$

Next we write $\tilde H_3$ in a moving frame, here $\tilde H_3$ is 
\begin{equation}\label{eq:H3}
\begin{array}{rcl}
&&\tilde H_{3}(\tilde{x_3},\tilde{y_3})=\frac{1}{2\pi}\log\left[\frac{\epsilon^2(\tilde x_1-\tilde x_3)^2+\epsilon^2(\tilde y_1-\tilde y_3)^2}{\epsilon^2(\tilde x_1-\tilde x_3)^2+(2+\epsilon(\tilde y_1+\tilde y_3))^2}\frac{\epsilon^2(\tilde x_2-\tilde x_3)^2+\epsilon^2(\tilde y_2-\tilde y_3)^2}{\epsilon^2(\tilde x_2-\tilde x_3)^2+(2+\epsilon(\tilde y_2+\tilde y_3))^2}\right]\\
=&&\frac{1}{2\pi}\log\frac{\epsilon^4}{16}+\frac{1}{2\pi}\log\{[(\tilde x_1-\tilde x_3)^2+(\tilde y_1-\tilde y_3)^2][(\tilde x_2-\tilde x_3)^2+(\tilde y_2-\tilde y_3)^2]\}\\
&&-\frac{\epsilon}{8\pi}(\tilde y_1+\tilde y_2+2\tilde y_3)+O(\epsilon^2).
\end{array}
\end{equation}

First we consider a frame moving with speed $1/2(\dot{\tilde {x_1}}+\dot{\tilde {x_2}})$, the new Hamiltonian in this frame is 
\begin{equation}\label{eq:H3prime}
H_3'(x',y')=\tilde H_{3}(\tilde{x_3},\tilde{y_3})-\frac 12 y'(\frac{\partial \tilde H_{12}}{\partial \tilde y_1}+\frac{\partial \tilde H_{12}}{\partial \tilde y_2})
\end{equation}
where $x'=\tilde x_3-1/2(\dot{\tilde {x_1}}+\dot{\tilde {x_2}})$, $y'=\tilde y_3$.

Then in a frame rotating with speed $\omega_0$, i.e., in coordinate $(x,y)$ given by $x+iy=(x'+iy')e^{i\omega_0t}$, the Hamiltonian becomes
$$
H(x,y,t)=H'(x',y')+\frac 12 \omega_0(x^2+y^2).
$$
Because $x+iy=(x'\cos (\omega_0t)+y'\sin (\omega_0t))+i(y'\cos (\omega_0t)-x'\sin (\omega_0t))$, the main term in $H$ is
$$
H_0(x,y)=\frac {1}{2\pi}\log[x^2+(y-1)^2][x^2+(y+1)^2]+\frac{1}{2}\omega_0(x^2+y^2).
$$

By Equation (\ref{eq:H_12}), $\frac{\partial \tilde H_{12}}{\partial \tilde y_1}+\frac{\partial \tilde H_{12}}{\partial \tilde y_2}=\frac{1}{2\pi}\cdot 2\epsilon(1+1)=2\epsilon/\pi$, hence in (\ref{eq:H3prime}), $1/2y'(\frac{\partial \tilde H_{12}}{\partial \tilde y_1}+\frac{\partial \tilde H_{12}}{\partial \tilde y_2})=-\epsilon y'/\pi=-\epsilon \tilde y_3/\pi$. Combining with (\ref{eq:H3}), the term of order $\epsilon$ is
$$
-\frac{\epsilon}{\pi}\tilde y_3-\frac{\epsilon}{8\pi}(\tilde y_1+\tilde y_2+2\tilde y_3)=-\frac{5}{4\pi}\epsilon\tilde y_3=-\epsilon\frac{5}{4\pi}(x\sin(\omega_0t)+y\cos(\omega_0t)).
$$ 

In conclusion, the Hamiltonian in the moving frame can be written as 
$$
H(x,y,t)=H_0(x,y)+\epsilon H_1(x,y,t)+O(\epsilon^2),
$$
where $H_0(x,y)=\frac {1}{2\pi}\log[x^2+(y-1)^2][x^2+(y+1)^2]+\frac{1}{2}\omega_0(x^2+y^2)$ and $H_1(x,y,t)=-5/(4\pi)(x\sin(\omega_0t)+y\cos(\omega_0t))$.

Note that up to the order of $\epsilon$, this system coincides with the perturbed system of restricted 3 vortex rings studied in \cite{BaBa}. Hence they have the same Melnikov integral. By the computation of the Melnikov integral in \cite{BaBa}, we can get that this restricted 3-vortex system in the half-plane is also nonintegrable. Then by the similar continuous argument used in \cite{Kha}, one can conclude that the 3-vortex system in the  half-plane is nonintegrable.

\end{proof}

\begin{remark}
The above proof shows similarity of the restricted 3-vortex system in the half-plane and the restricted system of 3 coaxial vortex rings. Two systems coincide up to the order of $\epsilon$, so they have similar properties. This is also true for the 2-vortex system in the half-plane and the system of 2 vortex rings. One can use this similarity to explain the reason that an interesting phenomenon of vortex pair in the half-plane, called leapfrogging (see cf. \cite{PTT}), is also observed for vortex rings.
\end{remark}

\medskip

\subsection{Point vortex system in the quadrant}\label{subsec:quadrant}

The Green function in the upper right quadrant is 
$$
G_{\mathbb R_+\times\mathbb R_+}({\bf x},{\bf x'})=-\frac {1}{2\pi}\log\|{\bf x}-{\bf x'}\|+\left[\frac {1}{2\pi}\log\|{\bf x}-{\bf x'^*}\|+\frac {1}{2\pi}\log\|{\bf x}+{\bf x'^*}\|-\frac {1}{2\pi}\log\|{\bf x}+{\bf x'}\|\right],
$$
where ${\bf x'^*}=(x',-y')$ is the mirror image of ${\bf x'}=(x',y')$. Note that the first term of $G_{\mathbb R_+\times\mathbb R_+}({\bf x},{\bf x'})$ is the Green function in $\mathbb R^2$, hence the next three terms are responsible for the self-interaction:
$$
\hat\gamma_{\mathbb R_+\times\mathbb R_+}({\bf x})=\gamma_{\mathbb R_+\times\mathbb R_+}({\bf x},{\bf x})=\frac{1}{2\pi}\log\frac{2xy}{\sqrt{x^2+y^2}},
$$
where ${\bf x}=(x,y)$ with $x\geq 0$ $y\geq 0$ is a point in the upper right quadrant.

Put the expression of $G_{\mathbb R_+\times\mathbb R_+}$ and $\hat\gamma_{\mathbb R_+\times\mathbb R_+}$ into the Hamiltonian function (\ref{Ham}), one can obtain the Hamiltonian function for the $N$-vortex system in the upper right quadrant $\mathbb R_+\times\mathbb R_+$.

For the motion of a single vortex in the upper right quadrant, the Hamiltonian function is $H_{\mathbb R_+\times\mathbb R_+}(x,y)=\frac{\Gamma^2}{2\pi}\log\frac{2xy}{\sqrt{x^2+y^2}}$, where $(x,y)$ is the position of  this single vortex and $\Gamma$ is its strength.

Hence the Hamiltonian equations for a single vortex moving in the quadrant are
\begin{equation}\label{eq:quadrant_Ham}
\left\{\begin{array}{l}
\dot{x}=\;\;\frac{1}{2\pi}\left(\frac 1 y-\frac{y}{x^2+y^2}\right), \\\\
\dot{y}=-\frac{1}{2\pi}\left(\frac 1 x-\frac{x}{x^2+y^2}\right).
\end{array}\right.
\end{equation}
 Since the Hamiltonian is conserved under the flow of (\ref{eq:quadrant_Ham}), the trajectories of a single vortex satisfy
$$
\frac{4x^2y^2}{x^2+y^2}=C^2,
$$
where $C$ is an arbitrary constant. In terms of polar coordinates, this trajectory equation can be written as $r=\frac{C}{sin2\theta}$. (A version of this equation can be found in the classical book \cite{Lamb} by Lamb.)
Also, one can get the following proposition from this explicit expression of the trajectory.
\begin{proposition}
The single vortex moving in the upper right quadrant will not hit the boundary.
\end{proposition}

Similar to the previous cases, one can ask the following natural questions.
\begin{question}\label{q4}
Is it possible that $N$-vortex system in the upper right quadrant collides or hits the boundary in  finite time when $N>1$? 
\end{question}
\begin{question}\label{q5}
Prove the nonintegrable for the $N$-vortex system in the upper right quadrant  when $N>1$.
\end{question}

\medskip

\section{Coaxial circular vortex rings}\label{sec:rings}

Consider the axisymmetric Euler equations without swirl  in $\mathbb R^3=(z,r,\theta)$ (``without swirl'' means the $\theta$ component of the velocity field vanishes in the cylindrical coordinates $(z,r,\theta)$) with
velocity and pressure depending only on $r$ and $z$:
\begin{equation}\label{eq:axisym}
\left\{\begin{array}{l}
\partial_t u_z+(u_z\partial_z+u_r\partial_r)u_z=-\partial_z p, \\
\partial_t u_r+(u_z\partial_z+u_r\partial_r)u_r=-\partial_r p, \\
\partial_zu_z+\frac 1 r \partial_r(ru_r)=0,
\end{array}\right.
\end{equation}
where $u=(u_z,u_r):\mathbb R_{+}^2 \rightarrow \mathbb R^2$ is the velocity field in the half-plane $\mathbb R_{+}^2$, which satisfies the boundary condition $u_r=0$ on $r=0$. Also note that the third equation in (\ref{eq:axisym}) above is equivalent to the equation $\nabla\cdot(ru)=\partial_z(ru_z)+\partial_r(ru_r)=0$.

One can see that Equations (\ref{eq:axisym}) are the lake equations (\ref{eq:lake}) in variables $(z,r)$ with $b(z,r)=r$ in the half-plane. (See Appendix \ref{sec:lakeeq} for an introduction to the lake equations. )

By Equation (\ref{eq:stream}) for the stream function of the lake equations, the elliptic PDE for the corresponding stream function $\psi$ (by definition $\psi$ satisfies $\nabla^{\perp}\psi=ru$, i.e. $\partial_z\psi=ru_r,\;\partial_r\psi=-ru_z$) of the axisymmetric Euler equation is  
\begin{equation}\label{eq:ring_stream}
-\partial^2_z\psi-\partial^2_r\psi+\frac 1 r \partial_r\psi=r\omega.
\end{equation}

\begin{remark}
If we consider the stationary solution of (\ref{eq:axisym}), the vorticity equation is 
$$u\cdot\nabla\left(\frac{\omega}{r}\right)=0,$$
since $\nabla^{\perp}\psi=ru$, one can conclude that $\frac{\omega}{r}$ is a function of $\psi$, i.e., $\frac{\omega}{r}=f(\psi)$, plug this into the RHS of (\ref{eq:ring_stream}), we obtain
\begin{equation}\label{eq:GS}
-\partial^2_z\psi-\partial^2_r\psi+\frac 1 r \partial_r\psi=r^2f(\psi),
\end{equation}
this is the famous Grad-Shafranov equation with vanishing swirl. For the general Grad-Shafranov equation with non-vanishing swirl, one can see \cite{GrRu,Sha}.
\end{remark}
\medskip

Now we look at the $N$-vortex (ring) problem for (\ref{eq:axisym}). Solving the elliptic PDE (\ref{eq:ring_stream}) gives us the stream function $\psi(z,r)$ at $(z,r)$:
\begin{equation}\label{eq:stream_ring}
\psi(z,r)=\frac {1}{4\pi}\int_0^{+\infty}\int_{-\infty}^{+\infty} rr'\omega(z',r')\;dz'\;dr'\int_0^{2\pi}\frac{\cos\theta\;d\theta}{\sqrt{(z-z')^2+r^2+r'^2-2rr'\cos\theta}},
\end{equation}
and the Green function for (\ref{eq:ring_stream}) is
\begin{equation}\label{eq:green_ring}
G(z,r,z',r')=\frac {rr'}{4\pi}\int_0^{2\pi}\frac{\cos\theta\;d\theta}{\sqrt{(z-z')^2+r^2+r'^2-2rr'\cos\theta}},
\end{equation}
these formulas can be found in \cite{Dyso} and \cite{Lamb}.

The speed $V$ of a thin cored vortex ring of circulation $\Gamma=\int\omega dz dr$, ring radius $R$ and core radius $a$ ($a/R\ll1$) is (see cf. \cite{Thom} and \cite{Lamb})
\begin{equation}\label{ring_speed}
V=\frac{\Gamma}{4\pi R}\left[\log\frac{8R}{a}-\frac 14 +O\left(\frac a R\right)\right].
\end{equation}
Since $\frac R a\gg 1$, one can treat $\log\frac {8R}{a}$ as a constant of order $\log\frac{1}{\epsilon}$, where $\epsilon\sim a/R$, i.e.,  the ring moves along the axis with a speed proportional to its curvature, which is a well-known result for vortex filaments.

When considering $N$ coaxial circular vortex rings, we need to count the motion due to the self-interaction, together with the interaction of different vortex rings. The $i$-th ring's position at time $t$ is given by the coordinate $(Z_i(t),R_i(t))$ of the center of its core, $i=1,2,\dots,N$, the equations in $(Z_i(t),R_i(t))$ are 
\begin{equation}\label{eq:ring}
\left\{\begin{array}{l}
\dot{Z_i}=\;\;\frac{\Gamma_i}{4\pi R_i}\left[\log\frac{8R_i}{a_i}-\frac 14 \right]+\frac{1}{\Gamma_iR_i}\frac{\partial U}{\partial R_i}, \\\\
\dot{R_i}=-\frac{1}{\Gamma_iR_i}\frac{\partial U}{\partial Z_i},
\end{array}\right.
\end{equation}
where $U$ is the interaction energy between different vortex rings
$$
U=\frac {1}{2\pi}\sum_{j\neq i}\Gamma_i\Gamma_j G(Z_i,R_i,Z_j,R_j),
$$
where $G(Z_i,R_i,Z_j,R_j)$ is the Green function defined above in (\ref{eq:green_ring}). This system is usually referred to as  Dyson's model in the literature \cite{Mele}.

The system (\ref{eq:ring}) can also be rewritten in the following Hamiltonian form,
\begin{equation}\label{eq:ring_Ham}
\left\{\begin{array}{l}
\Gamma_iR_i\dot{Z_i}=\;\;\frac{\partial H}{\partial R_i}, \\\\
\Gamma_iR_i\dot{R_i}=-\frac{\partial H}{\partial Z_i},
\end{array}\right.
\end{equation}
where the Hamiltonian function is $H=\sum\limits_{i=1}^N\frac{\Gamma_i^2}{4\pi}R_i\left[\log\frac{8R_i}{a_i}-\frac 74 \right]+U$. By taking the canonical variables $p_i=\Gamma_iR_i^2$ and $q_i=Z_i$, the system (\ref{eq:ring_Ham}) can be written in a canonical form.

Also, the system (\ref{eq:ring_Ham}) has two independent first integrals: $\sum\limits_{i=1}^N\Gamma_iR_i^2$ and the Hamiltonian $H$, hence we have the following theorem.
\begin{theorem}\label{thm:integ} The N-vortex ring problem for $N\leq 2$ is integrable. 
\end{theorem}

\begin{remark}
The nonintegrability of 3 coaxial vortex rings is proved in \cite{BaBa}. Up to the order $\epsilon$, the perturbed system in that proof appears to be the same as the perturbed system we used in Section \ref{subsec:half_plane}. This shows an interesting relation  of point vortex system in the half-plane and coaxial vortex rings.
\end{remark}

For higher-dimensional generalization of vortex ring, here are some natural open questions.
\begin{question}\label{q6}
One can consider a more general problem
of vortex spheres of codimension 2 in $\mathbb R^n$, centered on one axis, i.e., coaxial vortex spheres. The
corresponding Euler equation should reduce to the lake equations with
$b(z,r)=r^{n-2}$, where $r$ is the distance to the common axis. Prove that: 

1. a single vortex sphere travels along the axis with a speed proportional to its curvature,

2. integrability for 2 vortex spheres and nonintegrability for more than 2 vortex spheres,

3. Leapfrogging motion can happen for a pair of spheres.
\end{question}

\begin{question}\label{q6}
\footnote{The author would like to thank the anonymous reviewer for suggesting this question.}
More generally, consider the lake equations (\ref{eq:lake}) in variables $(z,r)$ with $b(z,r)=r^{\alpha},\;\alpha>0$ in the half-plane. Prove that: 

1. a single vortex travels parallel to the $z$-axis with a speed proportional to $1/r$,

2. integrability for the 2-vortex system and nonintegrability for more than 2 vortices,

3. leapfrogging motion can happen for a pair of vortices.
\end{question}

\medskip


\section{Sphere product vortex membranes}\label{sec:SphProd}

Assume that the Euler equation in $\mathbb R^{m+l+2}$ is sphere product $\mathbb S^m\times\mathbb S^l$-symmetric, i.e., the velocity $v$ and the pressure $p$ in the Euler equation are functions of the distances $(x,y)$ to the origin: $x=|X|, y=|Y|$ for $X\in \mathbb R^{m+1}, \, Y\in \mathbb R^{l+1}$, 
then the Euler equation can be written in the form of the lake equations (\ref{eq:lake}) with $b(x,y):=x^m y^l$.  

Now consider the motion of a sphere product vortex membrane under skew-mean-curvature flow (see Appendix \ref{sec:SMCF}). Here a sphere product vortex membrane is a singular vorticity $\xi=\delta_\Sigma$
supported on $\Sigma=\mathbb S^m(a)\times\mathbb S^l(b)\subset\mathbb R^{m+1}\times\mathbb R^{l+1}=\mathbb R^{m+l+2}$. It also can be regarded as the motion of a point vortex 
$\delta_{(a,b)}$ for $(a,b)\in \mathbb R_+\times\mathbb R_+$ for the corresponding lake equation.
The next theorem provides explicit solutions of point-vortex type, both existing forever or collapsing in finite time, depending on the membrane structure and dimension.

\begin{theorem}\label{prop:Sphere_product}(\cite{KhYa})
Let $F:\Sigma=\mathbb S^m(a)\times\mathbb S^l(b)\hookrightarrow\mathbb R^{m+1}\times\mathbb R^{l+1}=\mathbb R^{m+l+2}$
be the product of two spheres of radiuses $a$ and $b$. Then the evolution  $F_t$ of this surface $\Sigma$
in the binormal flow is the product of spheres 
$F_t(\Sigma)=\mathbb S^m(a(t))\times \mathbb S^l(b(t))$ at any $t$ with radiuses changing monotonically according to the ODE system:
\begin{equation}\label{eq:CliffordODE2}
\left\{\begin{array}{rcl}
\dot a&=&-l/b,\\
\dot b&=&+m/a.
\end{array}\right.
\end{equation}
For $0<m<l$ the corresponding solution $F_t$ exists only for finite time and collapses 
at $t=a(0)b(0)/(l-m)$. 
\end{theorem}

\begin{corollary}\label{cor:2}
In the general case of sphere products $\Sigma=\mathbb S^m(a)\times\mathbb S^l(b)$ the radiuses of $F_t(\Sigma)$
change as follows: $a(t)=ae^{-lt/(ab)}$ and $b(t)=be^{mt/(ab)}$ for $m=l$ and
$$
a(t)=a^{m/(m-l)}\left(a-(l-m)b^{-1}t\right)^{l/({l-m})}\;\text{and}\;\; b(t)=b^{l/(l-m)}\left(b+({m-l})a^{-1}t\right)^{m/(m-l)},
$$
for  $m\neq l$ and initial conditions $a(0)=a$, $b(0)=b$. 
\end{corollary}

\begin{remark}
This explicit solution might be useful to study the Euler singularity formation in higher dimensions, since the skew-mean-curvature flow is the localized induction approximation of the Euler equation.
The simplest case satisfying  the collapse condition $0<m<l$  is $m=1, l=2$ for  $ \mathbb S^1(a)\times\mathbb S^2(b)\subset \mathbb R^5$.  
Note also that the odd-dimensional Euler equation has fewer invariants (generalized helicities) than the even-dimensional one
(generalized enstrophies), see \cite{ArKh}. The existence of many invariants helps control solutions, so it indicates that 
the first example with a finite life-span occurs in the 5 dimensional case.
\end{remark}

\begin{proof}
For a point $q=(q_1,q_2)\in \mathbb S^m(a)\times\mathbb S^l(b)\hookrightarrow\mathbb R^{m+1}\times\mathbb R^{l+1}$, the  mean curvature of the sphere product is the vector
${\bf H}=-\frac ma {\bf n}_1-\frac lb {\bf n}_2$ (divided by the total dimension $m+l$ of the product, which we omit),
and the skew-mean-curvature vector is $-J{\bf H}=-\frac lb {\bf n}_1+\frac ma {\bf n}_2$, where ${\bf n}_1$ and ${\bf n}_2$ be the outer unit normal vectors to the corresponding spheres at the points $q_1$ and $q_2$ respectively.

The explicit form of the $-J{\bf H}$ vectors implies  the system of ODEs \eqref{eq:CliffordODE2} on the evolution of radiuses. Rewriting this as one first order ODE one can solve this explicitly, as in Corollary~\ref{cor:2}.
The system (\ref{eq:CliffordODE2}) is Hamiltonian on the $(a,b)$-plane with the Hamiltonian function given by $\mathcal H(a,b):=\ln(a^mb^l)$, which is the logarithm of the volume of the product of two spheres: ${\rm vol}(\Sigma)=C \,a^mb^l$.
(Note that the invariance of  this Hamiltonian is consistent with conservation of the volume 
of $\Sigma$, as the latter is the Hamiltonian  of the skew-mean-curvature flow.) 
\end{proof}

One can compare the binormal equation (\ref{eq:CliffordODE2}) of the sphere product vortex membrane with the
motion of a point vortex in the quadrant described by equations (\ref{eq:quadrant_Ham}).  Here are the graphs of their typical orbits.

\begin{figure}[H]
\centering
\includegraphics[width = 0.5\textwidth]{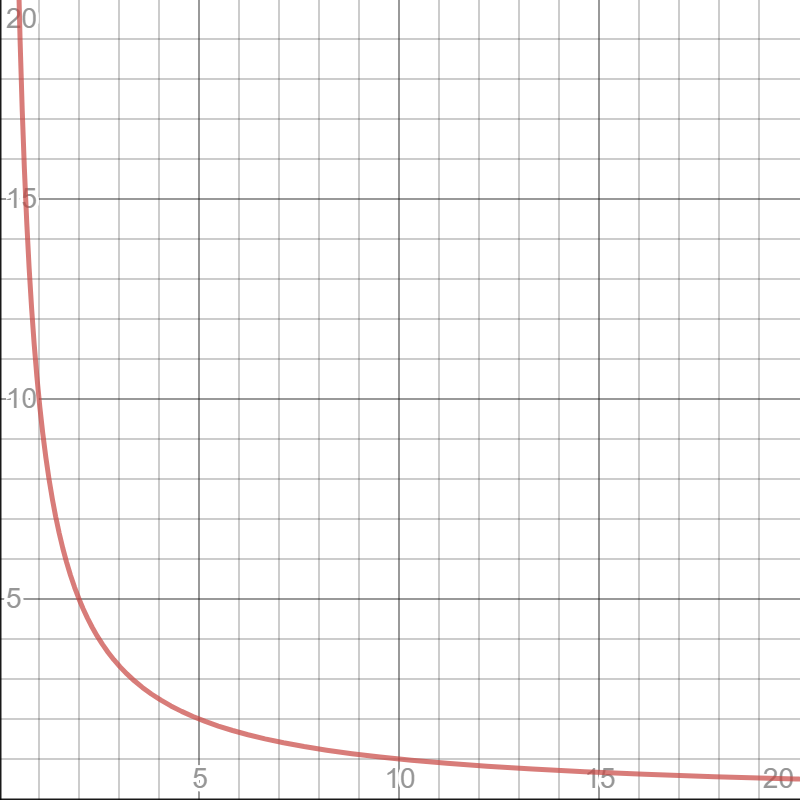}
\caption{An orbit of sphere product membrane under SMCF}
\end{figure}

\begin{figure}[H]
\centering
\includegraphics[width = 0.5\textwidth]{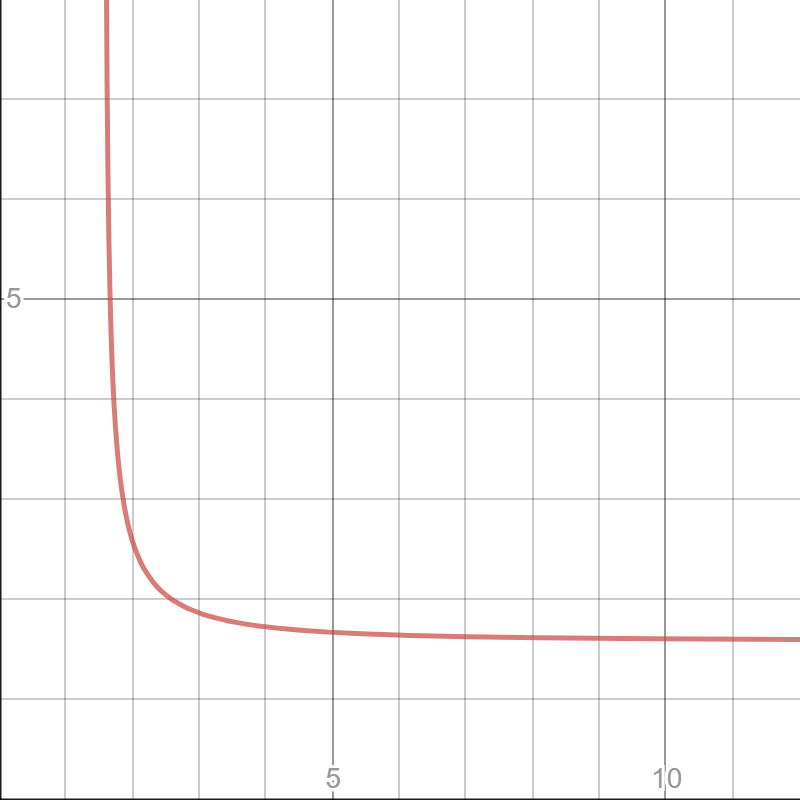}
\caption{An orbit of point vortex in the quadrant}
\end{figure}

Another difference is that the solution of (\ref{eq:quadrant_Ham}) exists for all time, while solutions of (\ref{eq:CliffordODE2}) exist for all times for $m=l$, but for $0<m<l$ it exists for a finite time only.

\medskip


\begin{appendices}
\section{The lake equations}\label{sec:lakeeq}

Let $D\subset\mathbb R^2$ be a planar domain and $b:D\rightarrow (0,+\infty)$ be a positive depth function, then 
the velocity field $v(t,\cdot)=(v_1,v_2): D \rightarrow \mathbb R^2$ and the surface height function $h(t,\cdot): D\rightarrow \mathbb R$ are governed by the lake equations:
\begin{equation}\label{eq:lake}
\left\{\begin{array}{l}
\partial_t v+(v,\nabla)v=-\nabla h, \;\;\; \text{on}\;D,\\
\nabla\cdot (bv)=0,\qquad\qquad\quad \text{on}\;D,\\
v\cdot {\bf n}=0,\qquad\qquad\qquad\text{on}\;\partial D,
\end{array}\right.
\end{equation}
where {\bf n} denotes the outgoing normal vector at the boundary $\partial D$ of the domain. See \cite{LOT} for more details about the lake equations.

Now take the vorticity function be $\omega=\partial_1v_2-\partial_2v_1$, we obtain the vorticity formulation of the lake equations (\ref{eq:lake}):
\begin{equation}\label{eq:vorlake}
\partial_t\left(\frac{\omega}{b}\right)+(v,\nabla)\left(\frac{\omega}{b}\right)=0.
\end{equation}
Then take $\psi$ be the stream function, i.e. $\nabla^{\perp}\psi=bv$, here $\nabla^{\perp}$ stands for $(\partial_2,-\partial_1)$, therefore we have 
$$
\omega=\partial_1v_2-\partial_2v_1=\partial_1(-\partial_1\psi/b)-\partial_2(\partial_2\psi/b),
$$ 
and simplifying this equation we obtain the following elliptic PDE:
\begin{equation}\label{eq:stream}
-\frac{\Delta \psi}{b}-\nabla\psi\cdot\nabla\left(\frac 1b\right)=\omega.
\end{equation}

\medskip

\section{Skew-mean-curvature flow}\label{sec:SMCF}

Skew-mean-curvature (or binormal) flows are localized approximations of the incompressible Euler equation in $\mathbb R^{n+2}$ with a singular vorticity profile supported on the membrane $\Sigma^n$, for codimension 2 vortex membranes in $\mathbb R^4$ see  \cite{Sh} and  in any dimension see \cite{HaVi, Jer, Kh}. Hence their blow-up/global existence results could shed some light on the motion of fluid flows themselves.

The skew-mean-curvature flow is defined as follows:

\begin{definition}\label{def:smcf}
Let $\Sigma^n\subset\mathbb R^{n+2}$ be a codimension  2  compact oriented submanifold (membrane) in the Euclidean space $\mathbb R^{n+2}$, the {\it skew-mean-curvature (or binormal) flow} is described by the equation:
\begin{equation}\label{eq:smcf}
\partial_t q =-J({\bf H}(q)),
\end{equation}
where $q\in \Sigma$, ${\bf H}(q)$ is the mean curvature vector at the point $q$ on $\Sigma$, $J$ stands for the operator of positive $\pi/2$ rotation in the two-dimensional normal space $N_q\Sigma$ to $\Sigma$ at $q$.

The skew-mean-curvature flow  (\ref{eq:smcf}) is a natural generalization of the binormal equation:
the mean curvature vector of a curve $\gamma$ in $\mathbb R^3$ at a point is ${\bf H}=\kappa\,\bf{n}$, where $\kappa$ is the curvature of the curve $\gamma$ at that point, hence the skew-mean-curvature flow becomes:
$\partial_t\gamma=-J(\kappa\,\bf{n}) = \kappa\,\bf{b}$, which is  the binormal equation for filaments. 
\end{definition}

On the  infinite-dimensional space  $\mathfrak M$ of codimension 2 membranes, there exists a natural symplectic structure:
\begin{equation}\label{eq:MW2}
\omega^{MV}(\Sigma)(u,v)=\int_{\Sigma}i_ui_v\mu,
\end{equation}
where $u$ and $v$ are two vector fields attached to the membrane $\Sigma\in\mathfrak M$, 
and $\mu$ is the volume form in $\mathbb R^{n+2}$. This is called the Marsden-Weinstein symplectic structure.

Let the functional ${\rm vol}(\Sigma)$ on the  space $\mathfrak M$ be  the $n$-dimensional volume of a compact $n$-dimensional membrane $\Sigma^n\subset \mathbb R^{n+2}$. Then
the skew-mean-curvature flow \eqref{eq:smcf} 
is the Hamiltonian flow on the membrane space $\mathfrak M$ equipped with the
Marsden-Weinstein structure whose Hamiltonian is given by the volume functional  ${\rm vol}$.

\end{appendices}

\end{document}